\documentclass[11pt]{article}
\usepackage{amsmath,amssymb,amsthm}
\usepackage{geometry}
\usepackage{authblk}
\usepackage{hyperref}
\usepackage{tikz}
\usepackage{pgfplots}
\usepackage[utf8]{inputenc}
\usepackage[T1]{fontenc}
\pgfplotsset{compat=1.18}
\usetikzlibrary{arrows.meta}
\usepackage{orcidlink}

\title{Hilbert's Sixth Problem and Soft Logic}
\author[1,2]{Moshe Klein}
\author[3]{Oren Fivel \orcidlink{0000-0002-2629-9907}}
\affil[1]{The University of Kiryat Shmona in the Galilee, Upper Galilee, 1220800, Israel (e-mail: kleinmos@telhai.ac.il)}
\affil[2]{The Hebrew University of Jerusalem, Mount Scopus, 9190500 Jerusalem, Israel
(e-mail: moshe.klein@mail.huji.ac.il)}
\affil[3]{Ben-Gurion University of the Negev, Beer Sheva 8410501, Israel (e-mail: fivel@post.bgu.ac.il}
\date{May 17, 2026}
\theoremstyle{plain}
\newcounter{theorem}
 \newtheorem{Theorem}[theorem]{Theorem}
 
 \newcounter{lemma}
 \newtheorem{Lemma}[lemma]{Lemma}
 
 \newcounter{corollary}
 \newtheorem{Corollary}[corollary]{Corollary}
\newcounter{axiom}
 \newtheorem{Axiom}[axiom]{Axiom}
\newcounter{definition}
 \newtheorem{Definition}[definition]{Definition}
 
 \newcounter{proposition}
 \newtheorem{Proposition}[proposition]{Proposition}
 
 \newcounter{characterization}
 \newcounter{property}
 \newcounter{problem}
 \newcounter{example}
 \newtheorem{Example}[example]{Example}
 
 \newcounter{examplesanddefinitions}
 \newcounter{remark}
 \newtheorem{Remark}[remark]{Remark}
 
 \newcounter{hypothesis}
 \newcounter{notation}
 \newcounter{assumption}
\providecommand{\keywords}[1]
{
  \small
  \textbf{\textit{Keywords---}} #1
}
\newcommand{\fb}{\perp}
\newcommand{\softplus}{\dot{+}}
\usepackage{xcolor}

\usepackage{graphicx} 
\usepackage{subfig}   

\usepackage{listings}
\definecolor{codegreen}{rgb}{0,0.6,0}
\definecolor{codegray}{rgb}{0.5,0.5,0.5}
\definecolor{codepurple}{rgb}{0.58,0,0.82}
\definecolor{backcolour}{rgb}{0.95,0.95,0.92}

\lstdefinestyle{mystyle}{
    backgroundcolor=\color{backcolour},   
    commentstyle=\color{codegreen},
    keywordstyle=\color{magenta},
    numberstyle=\tiny\color{codegray},
    stringstyle=\color{codepurple},
    basicstyle=\ttfamily\tiny,
    breakatwhitespace=false,         
    breaklines=true,                 
    captionpos=b,                    
    keepspaces=true,
    numbersep=5pt,                  
    showspaces=false,                
    showstringspaces=false,
    showtabs=false,                  
    tabsize=2
}

\begin{document}
\maketitle
\begin{abstract}
The Hilbert's sixth problem calls for the axiomatization of physics, particularly the derivation of macroscopic statistical laws from microscopic mechanical principles. A conceptual difficulty arises in classical probability theory: in continuous spaces every individual microstate has probability zero. In this paper, we introduce a probabilistic framework based on Soft Logic and Soft Numbers in which point events possess infinitesimal Soft probabilities rather than the classical zero. We show that Soft probability can be interpreted as an infinitesimal refinement of classical probability and discuss its implications for statistical mechanics and Hilbert's sixth problem. In addition, we show rigorously how to construct a Möbius-strip, based on the soft numbers, and we discuss how this Möbius-strip representation with soft numbers allows for a deeper understanding of the nature and character of Hilbert’s sixth problem. 

Inspired by the collapsing of that classical probability to zero, we suggest adding an axiom for an \emph{Infinitesimal Probability} into the list of Kolmogorov's five Probability axioms. Furthermore, we suggest a probabilistic framework based on Soft Numbers for assigning values to \emph{probabilities of impossible} events of a discrete random variable with realizations outside its support (which, in the ordinary probability, collapse to zero). This assignment of Soft Number values is based on an extension of the Pascal triangle to have soft zeros outside of the regular Pascal triangle (with real values) based on factorials of negative numbers.

\end{abstract}

\keywords{probability theory, Hilbert's sixth problem, soft logic, soft probability}

\section{Introduction} \label{sec: Intro}
In 1900, David Hilbert presented his sixth problem, calling for the axiomatization of physical theories, with particular emphasis on probability theory and its integration with geometric descriptions of physical phenomena \cite{Hilbert1902}. The problem highlights two main objectives: (i) to establish probability theory on rigorous axiomatic foundations, comparable to Hilbert's axiomatization of geometry, and (ii) to develop a geometric framework for physical laws, enabling precise mathematical descriptions of reality. A central challenge in statistical mechanics is the derivation of macroscopic statistical laws from microscopic deterministic dynamics.
A major step forward in solving Hilbert's sixth problem, occurred in 1933, when Andrey Kolmogorov established the axiomatic foundations of probability theory using measure theory, introducing the concept of a probability space \cite{Kolmogorov1933} (refer to \cite{Kolmogorov1956Foundation} for English version). In classical probability theory, defined axiomatically by Kolmogorov, the probability of a specific value of a continuous random variable is zero. However, physical systems occupy specific microstates in phase space.
The Hilbert's sixth problem has been discussed in the literature in order to resolve it. One recent example is by Deng et al. \cite{Deng2025} in which the authors rigorously derive the fundamental PDEs of fluid mechanics, such as the compressible Euler and incompressible Navier-Stokes-Fourier equations, starting from the hard sphere particle systems undergoing elastic collisions. Some other approaches to resolve Hilbert's sixth problem are the quantum probability (see e.g., Accardi's survey paper \cite{Accardi2018}).
The geometric framework development part of Hilbert's sixth problem is still an open issue.   
In 2016, Klein and Maimon \cite{Klein_Maimon2016a, Klein_Maimon2016b} have introduced Soft Logic, which makes a refinement of the number zero, by distinguishing between different multiples of zero organized along a zero axis (see also klein and Mimon's later works \cite{Klein_Maimon2019, Klein_Maimon2020, Klein_Maimon2021, Klein_Maimon2023} and their textbook \cite{Klein_Maimon2024}). This allows assigning infinitesimal Soft probabilities to microstates. In 2021, Fivel, Klein and Maimon \cite{FivelKleinMaimon2021} have proposed the Soft Probability, based on the soft logic, with application to decision trees (see also \cite{FivelKleinMaimon2023}). Herein, the authors developed a novel probability framework involving equality as a "soft zero" multiple of a probability density function (PDF).
In this work, we provide an overview of the results, definitions, and axioms of the soft logic and soft probability, and also provide a perspective of a solution to Hilbert's sixth problem from a point of view of soft logic and soft probability. 

\section{Classical Results in Probability Theory} \label{sec: Classical Results in Probability Theory}
In this section, we review key aspects of Kolmogorov’s probability theory. First we introduce Kolmogorov’s Classical axioms with motivation to add a new axiom. Then we introduce some properties of continuous random variables with motivation to introduce the soft logic-based soft probability definitions.

\subsection{Kolmogorov’s Classical Axioms}\label{ssec: Kolmogorov’s Classical Axioms}

Let $E$ be a sample space of elementary events $\xi, \eta, \zeta, \dots$, and let $\mathcal{F}$ be a collection of subsets of $E$ whose elements are called \textit{random events}. Let $\mathrm{Pr}(\cdot):\mathcal{F}\rightarrow[0,1]$ be the probability function. Kolmogorov’s original axioms are as follows (\cite{Kolmogorov1933}, p. 2, Chap. I, Sect. 1, Axioms I-V):

\begin{Axiom}[Field Structure]\label{axiom: kolmogorov Field Structure}
$\mathcal{F}$ is a field of sets, i.e., it is closed under complementary and finite unions.
\end{Axiom}

\begin{Axiom}[Containing the Sample Space]\label{axiom: kolmogorov Containing the Sample Space}
$\mathcal{F}$ contains $E$.
\end{Axiom}

\begin{Axiom}[Non-negativity]\label{axiom: kolmogorov Non-negativity}
To each set $A \in \mathcal{F}$ is assigned a non-negative real number $\mathrm{Pr}(A) \geq 0$. This number $\mathrm{Pr}(A)$ is called the probability of the event $A$.
\end{Axiom}

\begin{Axiom}[Normalization]\label{axiom: kolmogorov Normalization}
$\mathrm{Pr}(E) = 1$.
\end{Axiom}

\begin{Axiom}[Finite Additivity]\label{axiom: kolmogorov Finite Additivity}
If $A$ and $B$ have no element in common, then
\[
\mathrm{Pr}(A \cup B) = \mathrm{Pr}(A) + \mathrm{Pr}(B),
\]
where the symbol $\cup$ denotes the "union" (or equivalently the "addition" or the "or" statement) of the events $A$ and $B$.
\end{Axiom}

We suggest adding one more axiom (Axiom \ref{axiom: Infinitesimal Probability} below) for an Infinitesimal Probability:
\begin{Axiom}[Infinitesimal Probability]\label{axiom: Infinitesimal Probability}
An event $A$ is called \emph{infinitesimal} if its probability is smaller than any positive standard real number:
\[
0 < \mathrm{Pr}(A) < \varepsilon \quad \text{for all standard } \varepsilon > 0.
\]
\end{Axiom}
This requires extending the range of the probability measure to a non-Archimedean field containing infinitesimals. For infinite sample spaces, a sixth classical axiom of countable additivity (continuity) is typically added. The present paper extends the system in a different direction, replacing the classical zero-probability of point events by a soft-probabilistic refinement.

\subsection{Properties of continuous random variable}\label{ssec: Properties of continuous random variable}
In this subsection, we overview some properties of continuous random variables with motivation to introduce the soft logic-based soft probability definitions. A cumulative distribution function (CDF) $F_X$ of a continuous random variable $X$ is
\begin{equation} \label{eq: CDF def}
    F_X(x)=\textrm{Pr}(X\le x)=\Pr(X<x),
\end{equation}
so there is no distinction between strict and non-strict inequalities, i.e., the equality case collapses to zero ($\textrm{Pr}(X=x)=0$). Motivated by the following approximation, based on the PDF $f_X$
\begin{equation}\label{eq: Prob_approx}
\mathrm{Pr}(x<X\leq x+\Delta x) \approx f_X(x)\Delta x,  
\end{equation}
for some small $\Delta x>0$, and the following observation
\begin{equation}\label{eq: PR(X<=x) obs}
    \textrm{Pr}(X\le x)=0\cdot f_X(x)+1\cdot F_X(x),
\end{equation}
the probability of the random variable $X$ being less than or equal to a value $x$ is conceptually decomposed to a sum of a multiple of 1 (the CDF $F_X(x)$) and a multiple of 0 (the PDF $f_X(x)$, the CDF's derivative). Based on those observations in Eqs. \eqref{eq: Prob_approx}-\eqref{eq: PR(X<=x) obs}, the soft probability, has been developed in \cite{FivelKleinMaimon2021}, based on soft logic. 

\section{Soft Logic} \label{sec: Soft Logic}
The "Soft Logic" was developed by Klein and Maimon \cite{Klein_Maimon2016a}-\cite{Klein_Maimon2024} in order to present an extension of the number 0 from a singular point to a continuous line, and to distinguish between $-0$ and $+0$. A new type of numbers called "Soft numbers" have been developed, and has the following form:
\begin{equation}
a  \bar{0}\dot{+}b\bar{1},
\end{equation}
where $a, b \in \mathbb{R}$ such that $a$ denotes the multiples of $0$ in the $0$-axis (denoted by the symbol $\bar{0}$), and $b$ denotes the multiples of $1$ in the $1$-axis (denoted by the symbol $\bar{1}$, which will be omitted). The object $\bar{0}$ squares to zero (i.e., $\bar{0}^2=0\in\mathbb{R}$ but $\bar{0}\notin\mathbb{R}$).
In this section, we outline the axioms, definitions and results of soft logic. This outline is done in reverse, first, we outline the main definitions in soft probability, and then we continue to the fundamental of the soft logic, the soft numbers system with soft calculus and soft coordinate system. We continue with the Soft Logic by representation of the soft numbers on Möbius-strip.

\subsection{Soft Probability}\label{ssec: Soft Probability}
In this subsection, we define the notion of soft probability \cite{FivelKleinMaimon2021}. In order to distinguish from the ordinary probability operator $\textrm{Pr}(\cdot)$, we denote a new symbol for a soft probability operator by $\textrm{Pr}(\cdot)$, and we define it as follows with a continuous random variable $X$, and a real number $x$, that we compare to the continuous random variable. In \cite{FivelKleinMaimon2021} the notion of soft probability is extended to a soft probability of complements, unions and intersections, soft expectation, soft variance and soft entropy.     
\begin{Definition} \label{def: soft probability Ps}
    Let $\mathrm{Ps}(\cdot)$ be the soft probability operator. For a continuous random variable $X$ we define the soft probability of $X\leq x$ by
    \begin{equation}\label{eq: Soft_P}
        \begin{split}
        \mathrm{Ps}(X\leq x)&\overset{\textrm{def}}{=}\mathrm{Ps}(X=x)+\mathrm{Ps}(X<x)\\
        &\overset{\textrm{def}}{=}f_X(x)\bar{0}\dot{+}F_X(x),
        \end{split}
    \end{equation}
    where
    \begin{equation}\label{eq: Soft_P=af}
    \mathrm{Ps}(X=x)\overset{\textrm{def}}{=}f_X(x)\bar{0}, \\
    \end{equation}
    \begin{equation}\label{Soft_P<af}
    \mathrm{Ps}(X<x)\overset{\textrm{def}}{=}F_X(x) \equiv \mathrm{Pr}(X<x).
    \end{equation}
\end{Definition}
The result of the soft probability $\mathrm{Ps}(X\leq x)$ is a soft number, that is a sum of a term proportional to  1 (the CDF $F_X(x)$, for the strict inequality part) and a term proportional to 0 (the PDF $f_X(x)$, the CDF's derivative, for the equality part). The axioms and the definitions regarding the soft numbers are shown in the next subsection section.

\subsection{Soft Numbers} \label{ssec: Soft Numbers}
We provide the main results of the soft numbers based on \cite{Klein_Maimon2016a}-\cite{Klein_Maimon2024} for the completeness of the subject.  
\subsubsection{Axioms and Definitions} \label{sssec: Axioms and Definitions}
The following axioms and definitions are developed  for soft zeros for all real numbers $a$ and $b$:
\begin{Axiom}[Distinction]\label{axiom: SN Distinction} 
$a\neq b\Rightarrow a \bar{0}\neq b \bar{0}$.
\end{Axiom}
\begin{Definition}[Order] \label{def: SN Order} $a< b\Rightarrow a \bar{0}< b \bar{0}$.
\end{Definition}
\begin{Axiom}[Addition]\label{axiom: SN Addition}
$a \bar{0}+b \bar{0}=(a+b) \bar{0}$.
\end{Axiom}
\begin{Axiom}[Nullity]\label{axiom: SN Nullity}
$a \bar{0}\cdot b \bar{0}=0$, i.e., soft numbers "collapse" to zero under multiplications.
\end{Axiom}
\begin{Axiom}[Bridging] \label{axiom: Bridge No.}
There exists a bridge between a zero axis, and a real axis and vice versa, denoted by a pair of a bridge number and its mirror image about the bridge sign. Bridge numbers of a right type
\[
b \bar{1}\perp a \bar{0}
\]
and  bridge numbers of a left type 
\[
a \bar{0} \perp b \bar{1}.
\]
\end{Axiom}
\begin{Axiom}[Non-commutativity]\label{axiom: Bridge Not commute}
Bridging operator $\perp$ does not commute i.e.,  
\[b \bar{1}\perp a \bar{0} \neq a \bar{0} \perp b \bar{1}.\]
\end{Axiom}
\begin{Definition}[Soft Number]\label{def: SN def} \textit{A soft number is defined as a set of the of bridge number pairs of opposite types but with the same components -- the same zero axis number $a \bar{0}$ and the same real number $b$:}
\[
a \bar{0} \dot{+}b=\{a \bar{0}\perp b;b\perp a\bar{0}\},
\]
and we denote the set of the soft numbers by \textbf{\textup{SN}}.
\end{Definition}
From Axiom \ref{axiom: SN Nullity}, we have the following proposition:
\begin{Proposition}
  $a \bar{0}\cdot b \bar{0}+c=c\in\mathbb{R},\forall a,b,c\in\mathbb{R}$. 
\end{Proposition}
\noindent Based on the commutativity and the associativity of the real numbers under multiplication, we have the following proposition, with regards to scalar multiplication:  
\begin{Proposition}
    For all $a,b\in\mathbb{R}$, we have:
    \begin{enumerate}
        \item $a\bar{0}=\bar{0}a$,
        \item $b(a\bar{0})=a(b\bar{0})=(ab)\bar{0}$.
    \end{enumerate}
\end{Proposition}

\subsubsection{Mathematical operations and Functions on Soft Numbers} \label{sssec: math ops and funcs on SNs}
The following mathematical operations hold for the given soft numbers $a \bar{0} \dot{+}b, c \bar{0} \dot{+}d$:
\begin{itemize}
\item \textbf{Addition/subtraction:}
\begin{equation}\label{eq: softNumAdd}
(a\bar{0} \dot{+}b)\pm (c\bar{0} \dot{+}d)=(a\pm c)\bar{0} \dot{+}(b\pm d);
\end{equation}
\item \textbf{Multiplication:}
\begin{equation}\label{eq: softNumMult}
(a\bar{0} \dot{+}b)\cdot (c\bar{0} \dot{+}d)=(ad+bc)\bar{0} \dot{+}bd; 
\end{equation}
\item \textbf{Natural power:}
\begin{equation}\label{eq: softNumNaturalPwr}
(a\bar{0} \dot{+}b)^n=nab^{n-1}\bar{0} \dot{+}b^n.
\end{equation}
\item \textbf{Reciprocal (or Multiplicative Inverse):}
\begin{equation}
\begin{split}
    \frac{1}{a\bar{0} \dot{+}b}&=(a\bar{0} \dot{+}b)^{-1}\\&=-\frac{a}{b^2}\bar{0}\dot{+}\frac{1}{b}, b\neq0 
\end{split}   
\end{equation}
\end{itemize}
Based on the above equations, every polynomial $P_N(x)$ that operates on every soft number $\alpha \bar{0} \dot{+}x$ is given by
\begin{equation}\label{eq: softNumPoly}
P_N(\alpha \bar{0} \dot{+}x)=\alpha P_N'(x)\bar{0} \dot{+}P_N(x).
\end{equation}
where $P_N'(x)$ denotes the derivative of $P_N(x)$. This notion is generalized for analytic functions $f(x)$ so that
\begin{equation}\label{eq: softNumFunc}
f(\alpha \bar{0} \dot{+}x)=\alpha f'(x)\bar{0} \dot{+}f(x).
\end{equation}
In the context of \eqref{eq: softNumFunc}, the soft probability (Definition \ref{def: soft probability Ps}) forms a soft number with the CDF $F_X(x)$ as the function part that is a multiple of 1, and the PDF $f_X(x)=F_X'(x)$ as the derivative part that is a multiple of 0.

\subsubsection{Generalization to Polynomial Rings and Two-dimensional real algebras}
The set of the soft numbers \textbf{\textup{SN}} are defined by two real numbers: one in the real part (multiples of ones) and in the soft part (multiples of zeros). Yet, the notion of soft numbers can be generalized to more mathematical structures. For example, real vectors and real matrices (application examples: representation and assessment of Systems
Thinking competencies \cite{hirschprung2023} and risk evaluation in defense System of Systems \cite{hirschprung2026}). Other examples would be a ring of polynomials $\mathbb{R}[X]$, quotient ring etc.
\begin{Example}[Polynomial Ring]
Let $a_k,b_k,u_k,v_k\in\mathbb{R}$ for $k=0,1,...,n$. Then following holds as an incorporation of the soft numbers and the polynomial ring:
\begin{equation*}
    \sum_{k=0}^n(a_k \bar{0} \softplus b_k)X^k=\left( \sum_{k=0}^n {a_kX^k}\right)\bar{0} \softplus \left( \sum_{k=0}^n {b_kX^k}\right),
\end{equation*}
where the left hand side (LHS) is extension of the soft numbers into a polynomial ring ($\textbf{\textup{SN}}[X]$, polynomials with soft numbers coefficients), while the right hand side (RHS) is an extension of polynomial ring with real coefficients $\mathbb{R}[X]$ into a sum of two real polynomial ring with the form $\mathbb{R}[X] \bar{0} \softplus \mathbb{R}[X]$. Addition of such soft numbers' polynomial ring is done componentwize:
\begin{equation*}
    \sum_{k=0}^n(a_k \bar{0} \softplus b_k)X^k+\sum_{k=0}^n(u_k \bar{0} \softplus v_k)X^k=\left( \sum_{k=0}^n {(a_k+u_k)X^k}\right)\bar{0} \softplus \left( \sum_{k=0}^n {(b_k+v_k)X^k}\right).
\end{equation*}
For the multiplication, denote $A(x)=\sum_{k=0}^na_kX^k$, $B(x)=\sum_{k=0}^nb_kX^k$, $U(x)=\sum_{k=0}^nu_kX^k$ and $V(x)=\sum_{k=0}^nv_kX^k$, and observe that
\begin{equation*}\label{eq: soft ring mult}
\begin{split}
&\left(\sum_{k=0}^n(a_k \bar{0} \softplus b_k)X^k\right)\left(\sum_{k=0}^n(u_k \bar{0} \softplus v_k)X^k\right)=(A(X)\bar{0} \softplus B(X))(U(X)\bar{0} \softplus V(X))\\
&=(A(X)V(X)+B(X)U(X))\bar{0} \softplus B(X)V(X),
\end{split}
\end{equation*}
where the multiplications $A(X)V(X)$, $B(X)U(X)$ and $B(X)V(X)$ are done in the regular sense in the real polynomial ring to obtain new polynomials with a new degree less than or equal to $2n$. In the next example, we will on some quadratic algebra case (a special case of quotient ring), with some special cases of the complex numbers, the dual numbers and the split-complex numbers.
\end{Example}
 
\begin{Example}[Quadratic Algebra]
In this example, we focus on the quotient ring $\mathbb{R}[X]/(X^2-r)$, where the ideal $(X^2-r)$ satisfies $X^2=r$ for some $r\in\mathbb{R}$.\footnote{We do not assume collapsing $\varepsilon\bar{0}$ to zero.} 
 
Let $a_k,b_k,u_k,v_k\in\mathbb{R}$ for $k=0,1$. With $A=a_0+a_1X$, $B=b_0+b_1X$, $U=u_0+u_1X$ and $V=v_0+v_1X$, the following holds: 
\begin{equation*}
\begin{split}
    [&(a_0\bar{0}\softplus b_0)+(a_1\bar{0}\softplus b_1)X][(u_0\bar{0}\softplus v_0)+(u_1\bar{0}\softplus v_1)X]=[AV+BU]\bar{0} \softplus BV\\
    =&[(a_0+a_1X)(v_0+v_1X)+(b_0+b_1X)(u_0+u_1X)]\bar{0}\softplus[(b_0+b_1X)(v_0+v_1X)]\\
    =&[(a_0v_0+b_0u_0+[a_1v_1+b_1u_1]r)+(a_0v_1+a_1v_0+b_0u_1+b_1u_0)X]\bar{0}\\ &\softplus[b_0v_0+b_1v_1r+(b_1v_0+b_0v_1)X].
\end{split}
\end{equation*}
With a similar approach, we can calculate a natural power with the form $[(u_0\bar{0}\softplus v_0)+(u_1\bar{0}\softplus v_1)X]^n$ and likewise any polynomial $P_N(\cdot)$ of applied on $[(u_0\bar{0}\softplus v_0)+(u_1\bar{0}\softplus v_1)X]$. We focus on the general case of an analytic function $f$. Then,
\begin{equation*}
\begin{split}
f&((u_0\bar{0}\softplus v_0)+(u_1\bar{0}\softplus v_1)X)
=f((u_0+u_1X)\bar{0}\softplus(v_0+v_1X))\\
&=f(U \bar{0} \dot{+}V)=U f'(V)\bar{0} \dot{+}f(V)\\
&=(u_0+u_1X)f'(v_0+v_1X)\bar{0} \dot{+}f(v_0+v_1X),
\end{split}
\end{equation*}
where the derivative $f'(V)$ is calculated in such as for real functions. The followings are special cases of this Quadratic Algebra:
\begin{enumerate}
    \item \emph{Complex numbers (with $X=i$}): $U=u_0+u_1i$ and $V=v_0+v_1i$ are complex numbers, $f(V)=f_0(v_0,v_1)+if_1(v_0,v_1)$ is a complex differentiable function such that $f'(V)=\frac{\partial f_0}{\partial v_0}+i\frac{\partial f_0}{\partial v_1}=\frac{\partial f_1}{\partial v_1}-i\frac{\partial f_1}{\partial v_0}$;
    
    \item \emph{Split-complex numbers (with $X=j$}): $U=u_0+u_1j$ and $V=v_0+v_1j$ are split-complex numbers, $f(V)=f_0(v_0,v_1)+jf_1(v_0,v_1)$ is a split-complex differentiable function such that $f'(V)=\frac{\partial f_0}{\partial v_0}+j\frac{\partial f_0}{\partial v_1}=\frac{\partial f_1}{\partial v_1}+j\frac{\partial f_1}{\partial v_0}$; and

    \item \emph{Dual numbers (with $X=\varepsilon$}): $U=u_0+u_1\varepsilon$ and $V=v_0+v_1\varepsilon$ are dual numbers, $f(V)=f(v_0)+ f'(v_0)v_1\varepsilon$, $f'(V)=f'(v_0)+ f''(v_0)v_1\varepsilon$ and \[f(U \bar{0} \dot{+}V)=[u_0f'(v_0)+(u_0f''(v_0)v_1+u_1f'(v_0))\varepsilon]\bar{0} \dot{+}[f(v_0)+ f'(v_0)v_1\varepsilon],\]i.e., $f(U \bar{0} \dot{+}V)$ contains information on $f$, and on the first two derivatives $f'$ and $f''$.
\end{enumerate}
   
\end{Example}

\begin{Remark}
    A dual number $a\varepsilon+b$ have two possible representations as matrices in $M_2(\mathbb{R})$:
\begin{equation*}
\left\{  \begin{pmatrix}
b &  a\\
0 &  b\\
\end{pmatrix}, \begin{pmatrix}
b &  0\\
a &  b\\
\end{pmatrix}\right\}.    
\end{equation*}
Usually, only one representation is taken to be the equivalent matrix form of the dual numbers. In contrast, the soft number $a\bar{0}\softplus b$, as a set, includes two representations of numbers (the bridge numbers $a\bar{0}\fb b$ and $b\fb a\bar{0}$) with nilpotent components, the soft zero $a\bar{0}$, bridged from the left and the right of the real part $b$. In that sense, the soft numbers emphasize the two representations of combining a real number with a nilpotent part in a single set of matrices, rather than the traditional dual numbers.   
\end{Remark}

\subsubsection{Soft Coordinate Systems}\label{sssec: Soft Coordinate Systems}
We developed a new coordinate axis, as presented in Figure 1 below. It starts from zero to 1 horizontally and then it turns $90^\circ$ from 1 to infinity.

\begin{figure}[ht]
\centering
\begin{tikzpicture}[scale=1.2, line cap=round, line join=round]
    \draw[-{Latex[length=4mm,width=2.5mm]}, line width=1.2pt] (0,1) -- (0,4.5);
    \foreach \y/\lab in {1/1,2/2,3/3,4/4}{
        \draw[line width=1.2pt] (-0.18,\y) -- (0.18,\y);
        \node[right=10pt] at (0,\y) {\large $\lab$};
    }
    \draw[line width=1.2pt] (-2.6,1) -- (0,1);
    \fill (-2.3,1) circle (3.2pt);
    \fill (-1.1,1) circle (3.2pt);
    \draw[line width=1.2pt] (-1.1,1) -- (0,2);
    \node[below=6pt] at (-1.1,1) {\Large $\dfrac{1}{2}$};
    \node[below=28pt] at (-1.1,0.5) {\Large $x$};
    \node[right=18pt] at (0.55,2) {\Large $1/x$};
\end{tikzpicture}
\caption{The Soft coordinate axis (example of a line connected between $x=\frac{1}{2}$ and $\frac{1}{x}=2$)}
\label{fig: The Soft coordinate axis}
\end{figure}
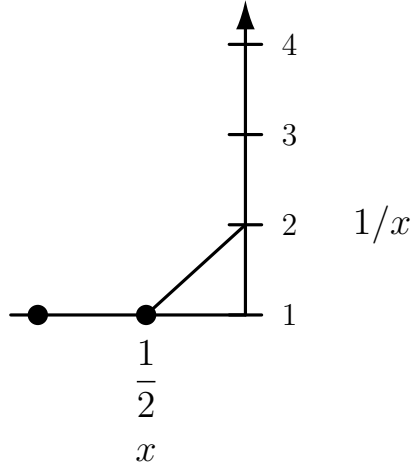

\begin{Lemma}\label{lemma: 1-1:(0,1]->[1,inf)}
There exists a one-to-one correspondence between the segment $(0,1]$ and the segment $[1,\infty)$.
\end{Lemma}
\begin{proof} Consider the following function $f(x)=\dfrac{1}{x}$. This function creates a one-to-one correspondence between $(0,1]$ and $[1,\infty)$.
\end{proof}

Figure \ref{fig: The Soft coordinate axis} contains an example of a line between $x=\frac{1}{2}$ and $\frac{1}{x}=2$. Also, one can draw such a line for all positive real $x$'s other than 1, from $x$ to $1/x$ as seen in Figure \ref{fig: lines between x and 1/x}.

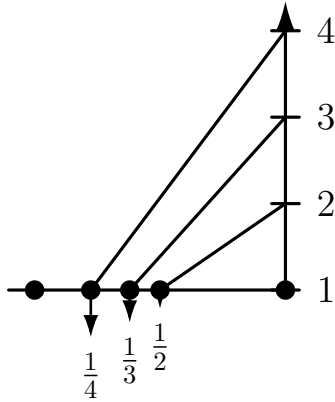
\begin{figure}[ht]
\centering
\begin{tikzpicture}[scale=1.15, line cap=round, line join=round]
    \draw[-{Latex[length=4mm,width=2.5mm]}, line width=1.2pt] (0,1) -- (0,4.35);
    \foreach \y/\lab in {1/1,2/2,3/3,4/4}{
        \draw[line width=1.2pt] (-0.16,\y) -- (0.16,\y);
        \node[right=8pt] at (0,\y) {\Large $\lab$};
    }
    \draw[line width=1.2pt] (-3.2,1) -- (0,1);
    \fill (-2.9,1) circle (3.2pt);
    \fill (-2.25,1) circle (3.2pt);
    \fill (-1.8,1) circle (3.2pt);
    \fill (-1.45,1) circle (3.2pt);
    \fill (0,1) circle (3.2pt);
    \draw[line width=1.2pt] (-2.25,1) -- (0,4);
    \draw[line width=1.2pt] (-1.8,1) -- (0,3);
    \draw[line width=1.2pt] (-1.45,1) -- (0,2);
    \draw[-{Latex[length=3mm,width=2mm]}, line width=1pt] (-2.25,1) -- (-2.25,0.45);
    \draw[-{Latex[length=3mm,width=2mm]}, line width=1pt] (-1.8,1) -- (-1.8,0.62);
    \draw[-{Latex[length=3mm,width=2mm]}, line width=1pt] (-1.45,1) -- (-1.45,0.78);
    \node[below=2pt] at (-2.25,0.45) {\Large $\frac{1}{4}$};
    \node[below=2pt] at (-1.8,0.62) {\Large $\frac{1}{3}$};
    \node[below=2pt] at (-1.45,0.78) {\Large $\frac{1}{2}$};
\end{tikzpicture}
\caption{The lines between $x=1/n$ (horizontal axis) and $1/x=n$ (vertical axis). We illustrate for some natural numbers $n$, but the connections lines are valid for all real numbers $n\ge1$.}
\label{fig: lines between x and 1/x}
\end{figure}

\begin{Lemma}
All the lines connecting $x$ to $1/x$ (for all nonzero real $x$) intersect at a single point.
\end{Lemma}
\begin{proof} Let us observe the following drawing with a line, which connects $x$ with $1/x$ as shown in Figure \ref{fig: The intersection point}.
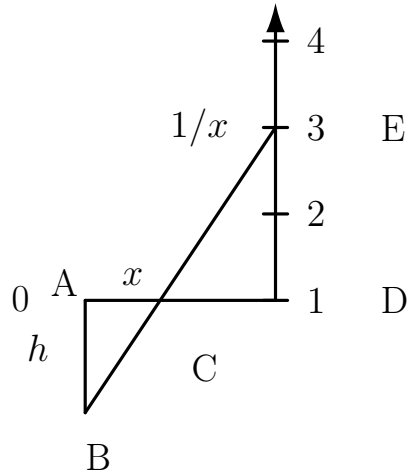
\begin{figure}[ht]
\centering
\begin{tikzpicture}[scale=1.15, line cap=round, line join=round]
    \draw[-{Latex[length=4mm,width=2.5mm]}, line width=1.2pt] (0,1) -- (0,4.45);
    \foreach \y in {1,2,3,4}{
        \draw[line width=1.2pt] (-0.14,\y) -- (0.14,\y);
    }
    \node[right=8pt] at (0,1) {\Large $1$};
    \node[right=8pt] at (0,2) {\Large $2$};
    \node[right=8pt] at (0,3) {\Large $3$};
    \node[right=8pt] at (0,4) {\Large $4$};
    \node[right=36pt] at (0,1) {\Large D};
    \node[right=36pt] at (0,3) {\Large E};
    \draw[line width=1.2pt] (-2.2,1) -- (0,1);
    \draw[line width=1.2pt] (-2.2,1) -- (-2.2,-0.3);
    \draw[line width=1.2pt] (-2.2,-0.3) -- (0,3);
    \node[left=12pt] at (-1.8,1.2) {\Large A};
    \node[below=8pt] at (-2.05,-0.3) {\Large B};
    \node[below right=2pt] at (-1.15,0.55) {\Large C};
    \node[above=2pt] at (-1.65,1) {\Large $x$};
    \node[left=10pt] at (-0.12,3) {\Large $1/x$};
    \node[left=10pt] at (-2.2,0.45) {\Large $h$};
    \node[left=10pt] at (-2.42,1.02) {\Large 0};
\end{tikzpicture}
\caption{The intersection point}
\label{fig: The intersection point}
\end{figure}
$\triangle ABC \sim \triangle CDE$ as triangles with equal angles. Consequently,
\[
\frac{AC}{AB}=\frac{CD}{ED}
\]
and therefore,
\[
\frac{x}{h}=\frac{1-x}{\frac{1}{x}-1}
\]
and hence
\[
h=1.
\]
This means that all the lines that connect $x$ to $1/x$ intersect at the same point, which is located one unit below zero.
\end{proof}

The intersection point denotes the beginning of the Soft coordinate system. We call this point ``absolute zero.'' The distance from absolute zero to $+\bar{0}$ is one unit. We suggest extending this new coordinate axis symmetrically as shown in Figure \ref{fig: distinction between -0 and +0}.
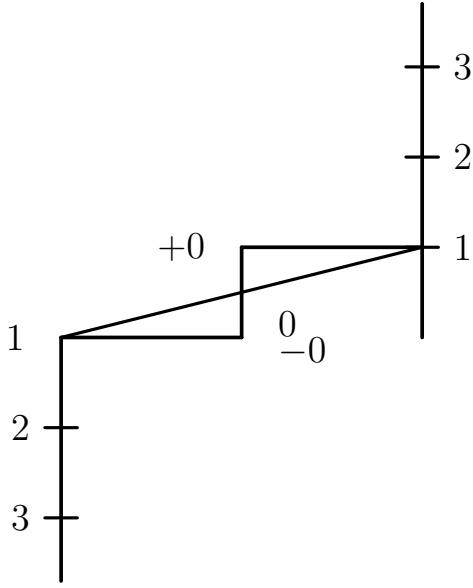
\begin{figure}[ht]
\centering
\begin{tikzpicture}[scale=1.2, line cap=round, line join=round]
    \draw[line width=1.4pt] (2,1) -- (2,4.7);
    \draw[line width=1.4pt] (-2, -1.7) -- (-2,1);
    \foreach \y/\lab in {2/1,3/2,4/3}{
        \draw[line width=1.2pt] (1.82,\y) -- (2.18,\y);
        \node[right=8pt] at (2,\y) {\Large $\lab$};
    }
    \foreach \y/\lab in {0/2,-1/3}{
        \draw[line width=1.2pt] (-2.18,\y) -- (-1.82,\y);
        \node[left=8pt] at (-2,\y) {\Large $\lab$};
    }
    \draw[line width=1.4pt] (-2,1) -- (0,1);
    \draw[line width=1.4pt] (0,1) -- (0,2);
    \draw[line width=1.4pt] (0,2) -- (2,2);
    \draw[line width=1.2pt] (-2,1) -- (2,2);
    \node[left=10pt] at (-2,1) {\Large $1$};
    \node[left=10pt] at (0,2) {\Large $+0$};
    \node[right=10pt] at (0,1.15) {\Large $0$};
    \node[right=10pt] at (0,0.85) {\Large $-0$};
\end{tikzpicture}
\caption{The distinction between $-0$ and $+0$}
\label{fig: distinction between -0 and +0}
\end{figure}
We now have three zeros instead of one zero. One zero is opposite to the number -1, and it is not identical to the zero opposite the number $+1$. Hence, we suggest denoting the two different ``zeros'' as $-\bar{0}$ and $+\bar{0}$. Furthermore, the discovery of the zero line enables us to find the locations of the multiples of Soft zero as.
The zero line is developed according to the three axioms presented above. Figure \ref{fig: complete Soft coordinate system} presents the extended new coordinate system for positive and negative numbers with the multiples of Soft zero.
\begin{figure}[ht]
\centering
\begin{tikzpicture}[scale=0.95, line cap=round, line join=round]
    \def\xL{-2.2}
    \def\xM{0}
    \def\xR{2.2}
    \draw[line width=1.2pt] (\xL,-5) -- (\xL,5);
    \draw[line width=1.2pt] (\xM,-5) -- (\xM,5);
    \draw[line width=1.2pt] (\xR,-5) -- (\xR,5);
    \draw[line width=1.2pt] (\xL,5) -- (\xR,5);
    \draw[line width=1.2pt] (\xL,-5) -- (\xR,-5);
    \draw[line width=1.2pt] (\xL,0) -- (\xR,0);
    \draw[dashed, line width=0.9pt] (\xL,-1) -- (\xR,1);
    \draw[dashed, line width=0.9pt] (\xL,1) -- (\xR,-1);
    \foreach \y in {-5,-4,-3,-2,-1,0,1,2,3,4,5}{
        \fill (\xL,\y) circle (2.8pt);
    }
    \foreach \y in {-5,-4,-3,-2,-1,1,2,3,4,5}{
        \fill (\xM,\y) circle (2.8pt);
    }
    \fill (\xM,0) circle (2.8pt);
    \foreach \y in {-5,-4,-3,-2,-1,1,2,3,4,5}{
        \fill (\xR,\y) circle (2.8pt);
    }
    \fill (-1.1,5) circle (2.8pt);
    \fill ( 1.1,5) circle (2.8pt);
    \fill (-1.1,-5) circle (2.8pt);
    \fill ( 1.1,-5) circle (2.8pt);
    \foreach \y/\lab in {
        5/-5, 4/-4, 3/-3, 2/-2, 1/-1,
        -1/1, -2/2, -3/3, -4/4, -5/5
    }{
        \node[left=8pt] at (\xL,\y) {\Large $\lab$};
    }
    \node[left=8pt] at (\xM, 5.3) {\Large $5\!\cdot\!\bar{0}$};
    \node[left=8pt] at (\xM, 4) {\Large $4\!\cdot\!\bar{0}$};
    \node[left=8pt] at (\xM, 3) {\Large $3\!\cdot\!\bar{0}$};
    \node[left=8pt] at (\xM, 2) {\Large $2\!\cdot\!\bar{0}$};
    \node[left=8pt] at (\xM, 1) {\Large $1\!\cdot\!\bar{0}$};
    \node[left=8pt] at (\xM,-1) {\Large $-1\!\cdot\!\bar{0}$};
    \node[left=8pt] at (\xM,-2) {\Large $-2\!\cdot\!\bar{0}$};
    \node[left=8pt] at (\xM,-3) {\Large $-3\!\cdot\!\bar{0}$};
    \node[left=8pt] at (\xM,-4) {\Large $-4\!\cdot\!\bar{0}$};
    \node[left=8pt] at (\xM,-5.3) {\Large $-5\!\cdot\!\bar{0}$};
    \node[left=8pt] at (\xM,0.26) {\Large $0\!\cdot\!\bar{0}$};
    \foreach \y/\lab in {
        5/5, 4/4, 3/3, 2/2, 1/1,
        -1/-1, -2/-2, -3/-3, -4/-4, -5/-5
    }{
        \node[right=8pt] at (\xR,\y) {\Large $\lab$};
    }
    \draw[latex-latex,dashed] (1.1,4.9)--(1.1,0.0);
    \draw[latex-latex,dashed] (0.0,4.5)--(1.1,4.5);
    \node[right=10pt] at (0.3,2.5) {\colorbox{white}{\Large $A$}};
    \node[right=10pt] at (-0.25,4.2) {\colorbox{white}{\Large $B$}};
    \node[above=6pt] at (1.2,5.1) {\Large $C$};

\end{tikzpicture}
\caption{The complete Soft coordinate system. The soft number point $C$ is presented by the height $A$ (=5) and width $B$ ($\approx$0.5).}
\label{fig: complete Soft coordinate system}
\end{figure}
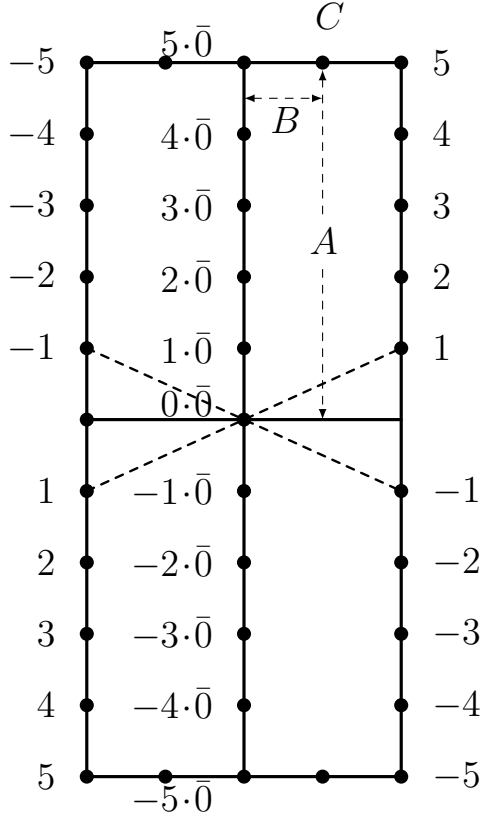

In Figure \ref{fig: complete Soft coordinate system}, the middle vertical line is the zero axis that includes all the multiples of 0: $x\bar{0}$. The right-hand vertical line is the 1 axis that includes all the multiples of 1: $y\bar{1}$. The left-hand vertical line is the 1 axis that includes all the multiples of 1: $y\bar{1}$ but in the opposite direction. This coordinate system allows us to present all of the real numbers and all of the Soft zeros. We can now use this coordinate system for representing various Soft numbers. Given a Soft number, where $x$ and $y$ are any real numbers, we represent both elements of a Soft number as two points, on either side of a page. The point $x\bar{0}\fb y$ appears on one side of the page and the point $y\fb x\bar{0}$ is the same point but on the other side of the page. Thus,
\[
c=x\bar{0}\fb y,
\]
is located on the front side of the page, while
\[
c'=y\fb x\bar{0}
\]
is located on the back side of the page. As the infinite strip presented (partially) in Figure \ref{fig: Soft Möbius Map} is intended for the presentation of Soft numbers, we call it a ``Soft number strip'' (SNS).

\begin{Definition} \label{def: point on SNS}
 Let $C$ be any point on the SNS. The height of point $C$ is the vertical distance from $C$ to the horizontal segment with the absolute zero at its center. $A$ denotes the height with a sign. The width of point $C$ is the horizontal distance from $C$ to the zero line with a sign and is denoted by $B$.   
\end{Definition}
 
The notations in Definition \ref{def: point on SNS} provide every point $C$ on the SNS with two parameters, $A$ and $B$, where $-\infty<A<\infty,\ -1\le B\le 1$. The condition $A>0$ is satisfied in the positive part of the SNS; $A<0$, in its negative part, or correspondingly, above and below the horizontal segment containing the absolute zero, while on this segment $A=0$.
For the second parameter, $B$, there is $B=0$ on the zero axis, $B=1$ on the right line bounding the SNS, $B=-1$ on the left line bounding the SNS and otherwise $-1<B<1$.
Given $A\in R,\ B\in[-1,1]$, we define $C(A,B)$ as follows:
For $0\le B\le 1,$
\[
c=(1-B)A\bar{0}\fb BA\bar{1},
\]
located on the front side of the page
\[
c'=BA\bar{1}\fb (1-B)A\bar{0},
\]
located on the back side of the page,
\[
C=x\bar{0}\softplus y=\{c,c'\}
\]
For $-1\le B\le 0$
\[
c=(1+B)A\bar{0}\fb BA\bar{1},
\]
located on the front side of the page
\[
c'=BA\bar{1}\fb (1+B)A\bar{0},
\]
located on the back side of the page
\[
x\bar{0}\softplus y=\{c,c'\}
\]
Note that if we use $|B|$, the two equations above are the same.
So, in brief, and in terms of the absolute value of $B$:
\[
c=(1-|B|)A\bar{0}\fb BA\bar{1},
\]
located on the front side of the page,
\[
c'=BA\bar{1}\fb (1-|B|)A\bar{0},
\]
located on the back side of the page, and therefore
\[
x\bar{0}\softplus y=\{c,c'\}
\]
Hence, by a coefficient comparison of the real part and the Soft part in
\begin{equation}\label{eq: AB to xy}
\begin{split}
  x&=(1-|B|)A,\\
  y&=BA,
\end{split}    
\end{equation}
or equivalently, after solving the above, $(A,B)$, using some arithmetic of absolute value and sign functions (we use the multivalued sign function so that $\operatorname{sign}(0)=\{\pm1\}$):
\begin{equation}\label{eq: xy to AB}
\begin{split}
    A&=(|x|+|y|)\operatorname{sign}(x),\\
    B&=\frac{y\cdot \operatorname{sign}(x)}{|x|+|y|}.
\end{split}    
\end{equation}

We show that Eqs. \eqref{eq: AB to xy}-\eqref{eq: xy to AB} are inverse functions of each other for $x\neq 0$ \textbf{and} $y\neq 0$ \textbf{and} $A\neq 0$ \textbf{and} $B\neq 0$ \textbf{and}  $B\neq 1$ \textbf{and} $B\neq -1$. In addition, we explore the relations between Eqs. \eqref{eq: AB to xy}-\eqref{eq: xy to AB} when $x,y,A=0;B=0,\pm1$.

Recall that a function $f:X \to Y$ is \emph{invertible} iff it is \emph{bijective} iff it is \emph{injective (one to one)} and \emph{surjective (onto)}.
A function $f:X \to Y$ is injective if for all $x_1,x_2 \in X$, if $f(x_1)=f(x_2)$, then $x_1=x_2$. A function $f:X \to Y$ is surjective if for all $y\in Y$ there is some $x\in X$ (not necessarily unique) such that $y=f(x)$ (if unique, then bijective).

We will prove that Eqs. \eqref{eq: AB to xy}-\eqref{eq: xy to AB} form a bijection of each other. We denote the RHS of Eq. \eqref{eq: AB to xy} by $f(A,B)$ (Eq. \eqref{eq: f(A,B)} below), and the RHS of Eq. \eqref{eq: xy to AB} by $g(x,y)$ (Eq. \eqref{eq: g(x,y)} below). First we will prove in Therorem \ref{theorem: SN-SP bijective} that $f$ and $g$ are bijective, and then we will prove in Corollary \ref{corrollary: f and g inverse ea other} that $f$ and $g$ are inverse functions of each other.  

\begin{Theorem}\label{theorem: SN-SP bijective}
    Suppose that $x,y,A,B\in\mathbb{R}$ such that $B\in[-1,1]$. Let
    \begin{alignat}{2}\label{eq: f(A,B)}
        f(A,B)=((1-|B|)A,BA) && \qquad \text{(cf. \eqref{eq: AB to xy})},
    \end{alignat}
    and
    \begin{equation}\label{eq: g(x,y)}
        g(x,y)=(\mathrm{sign}(x)(|x|+|y|),\mathrm{sign}(x)\frac{y}{|x|+|y|}) \text{ (cf. \eqref{eq: xy to AB})}.
    \end{equation}
    If $x\neq 0$ \textbf{and} $y\neq 0$ \textbf{and} $A\neq 0$ \textbf{and} $B\neq 0$ \textbf{and}  $B\neq 1$ \textbf{and} $B\neq -1$, then $f$ and $g$ are bijective functions.
\end{Theorem}
\begin{proof}
    The proof is accomplished by the following steps, required to be proven:
    \begin{enumerate}
        \item $f$ is \emph{injective};
        \item $f$ is \emph{surjective};
        \item $f$ is \emph{injective} and \emph{surjective} $\Rightarrow$ $f$ is \emph{bijective};
        \item $g$ is \emph{injective};
        \item $g$ is \emph{surjective};
        \item $g$ is \emph{injective} and \emph{surjective} $\Rightarrow$ $g$ is \emph{bijective};
        \item $f$ and $g$ are \emph{bijective}.
    \end{enumerate}
    Steps 1 and 2 will require computational work, while step 3 will be immediately true from steps 1 and 2. Steps 4 and 5 will require computational work, while step 6 will be immediately true from steps 4 and 5. Step 7 will be immediately true from steps 3 and 6. 
    
    \noindent\underline{\emph{Step 1}}: Suppose that $f(A,B)=f(a,b)=(x,y)$ for some $x,y,A,a\neq0$ and $B,b\in(-1,1)\backslash\{0\}$. Therefore
    \[(1-|B|)A=(1-|b|)a=x,BA=ba=y \]
    Since all factor are non-zero we can write $a=\frac{BA}{b}$and $b=\frac{BA}{a}$. Using the property $r=|r|\mathrm{sign}(r)$, an implicit solution for $a$ in terms  of $x,y,A,B$ and $\mathrm{sign}(a)$ after plugging $b=\frac{BA}{a}$ is 
    \[
    a=A+|BA|(\mathrm{sign}(a)-\mathrm{sign}(A))=x+|y|\mathrm{sign}(a),
    \]
    and an implicit solution for $b$ in terms  of $x,y,A,B$ and $\mathrm{sign}(b)$ after plugging $a=\frac{BA}{b}$ is
    \[
    b=\left( \frac{1}{1+[\mathrm{sign}(b)-\mathrm{sign}(B)]}\right)B=\frac{y}{x+y\mathrm{sign}(b)}.
    \]
    Since we assume that $B,b\in(-1,1)\backslash\{0\}$, we have
    the following signum equalities
    \begin{align*}
     \mathrm{sign}(a)&=\mathrm{sign}(A)=\mathrm{sign}(x)\\
     \mathrm{sign}(b)\mathrm{sign}(a)&=\mathrm{sign}(B)\mathrm{sign}(A)=\mathrm{sign}(y)\\
     \mathrm{sign}(b)&=\mathrm{sign}(B)=\mathrm{sign}(y)\mathrm{sign}(x),\\
    \end{align*}
 and more importantly we have $\mathrm{sign}(a)-\mathrm{sign}(A)=\mathrm{sign}(b)-\mathrm{sign}(B)=0$, which implies (after factoring $\mathrm{sign}(x)$)
 \begin{align}
     a&=A=x+|y|\mathrm{sign}(x)=\mathrm{sign}(x)(|x|+|y|)\\
     b&=B=\frac{y}{x+y\mathrm{sign}(y)\mathrm{sign}(x)}=\mathrm{sign}(x)\frac{y}{|x|+|y|}.
 \end{align}
We have shown that if $f(A,B)=f(a,b)$ then $(A,B)=(a,b)$ and thus $f$ is \emph{injective} (step 1 proven). 

\noindent\underline{\emph{Step 2}}: We have $(A,B)=(a,b)=g(x,y)$. Denote $a=a(x,y)$ and $b=b(x,y)$ as the above solution for $(a,b)$ in terms of $(x,y)$, and observe that $a(x,y)\in\mathbb{R}\backslash\{0\}$ and $b(x,y)\in(-1,1))\backslash\{0\}$. For all $(x,y)$ such that $x,y\neq0$ the pair $(a(x,y),b(x,y))$ satisfies the surjectivity of $f$ because
\begin{equation}
\begin{split}
 f(A,B)&=f(a(x,y),b(x,y))\\
 &=((1-|b(x,y)|)a(x,y),b(x,y)a(x,y))\\
 &=\left(\left(\mathrm{sign}(x)-\frac{|y|\mathrm{sign}(x)}{|x|+|y|}\right)(|x|+|y|),y[\mathrm{sign}(x)]^2\frac{|x|+|y|}{|x|+|y|}\right)\\
 &=(x,y).
\end{split}
\end{equation}
Therefore we have proved that for all $(x,y)\in\mathbb{R}\backslash\{0\}\times\mathbb{R}\backslash\{0\}$ there is $(A,B)\in\mathbb{R}\backslash\{0\}\times(-1,1)\backslash\{0\}$ such that $f(A,B)=(x,y)$, and thus $f$ is \emph{surjective} (step 2 proven).

\noindent\underline{\emph{Step 3}}: $f$ is \emph{injective} (from step 1) and \emph{surjective} (from step 2), and thus $f$ is \emph{bijective} (step 3 proven).

Now we will sketch a proof that $g$ is bijective (it is the reverse calculations for the proof of $f$ to be bijective).

\noindent\underline{\emph{Step 4}}: Suppose that $g(x,y)=g(X,Y)=(A,B)$ for some $x,y,X,Y,A\neq0$ and $B\in(-1,1)\backslash\{0\}$. Therefore
    \[\mathrm{sign}(X)(|X|+|Y|)=\mathrm{sign}(x)(|x|+|y|)=A,\]
    \[\mathrm{sign}(X)\frac{Y}{|X|+|Y|}=\mathrm{sign}(x)\frac{y}{|x|+|y|}=B \]
We can express implicitly $(X,Y)$ as follows
\begin{align*}
    X&=x+[\mathrm{sign}(x)-\mathrm{sign}(X)]|BA|=A-|BA|\mathrm{sign}(X)\\
    Y&=y=BA.
\end{align*}
For similar reasons from previous part of the proof, we have the signum equality $\mathrm{sign}(x)=\mathrm{sign}(X)=\mathrm{sign}(A)$, thus we have an explicit solution for $(X,Y)$
\begin{align}
    X&=x=(1-|B|)A\\
    Y&=y=BA.
\end{align}
$g$ is \emph{injective} because $g(x,y)=g(X,Y)$ implies $(x,y)=(X,Y)$ (step 4 proven).

\noindent\underline{\emph{Step 5}}: We also have $(x,y)=(X(A,B),Y(A,B))=f(A,B)$, where $(X(A,B),Y(A,B))$ is the solution of $(X,Y)$ in terms of $(A,B)\in\mathbb{R}\backslash\{0\}\times(-1,1)\backslash\{0\}$, which satisfies the surjectivity condition.

\begin{equation}
\begin{split}
    g(x,y)&=g((X(A,B),Y(A,B))\\
    &=g((1-|B|)A,BA)|_{\mathrm{sign}((1-|B|)A)=\mathrm{sign}(A)}\\
    &=(\mathrm{sign}(A)(|(1-|B|)A|+|BA|),\mathrm{sign}(A)\frac{BA}{|(1-|B|)A|+|BA|})\\
    &=(|A|\mathrm{sign}(A)(\left|1-|B|\right|+|B|),B\frac{\mathrm{sign}(A)A/|A|}{\left|1-|B|\right|+|B|})|_{\left|1-|B|\right|=1-|B|}\\
    &=(A,B).
\end{split}    
\end{equation}
For all $(A,B)\in\mathbb{R}\backslash\{0\}\times(-1,1)\backslash\{0\}$, there is $(x,y)\in\mathbb{R}\backslash\{0\} \times\mathbb{R}\backslash\{0\}$ such that $g(x,y)=(A,B)$, and thus $g$ is \emph{surjective} (step 5 proven).

\noindent\underline{\emph{Step 6+7}}: Since $g$ \emph{injective} (from step 4) and \emph{surjective} (from step 5), $g$ is  \emph{bijective} (step 6 proven). Therefore, $f$ (from step 3) and $g$ (from step 6) are  \emph{bijective} (step 7 proven).
\end{proof}

\begin{Corollary}\label{corrollary: f and g inverse ea other} 
 $f(A,B)$  and $g(x,y)$ are inverse of each other, for all $(x,y)\in(x,y)\in\mathbb{R}\backslash\{0\} \times\mathbb{R}\backslash\{0\}$ and  $(A,B)\in\mathbb{R}\backslash\{0\}\times(-1,1)\backslash\{0\}$.   
\end{Corollary}
\begin{proof}
Through the above proof of Theorem \ref{theorem: SN-SP bijective}, we have shown that if $(x,y)\in(x,y)\in\mathbb{R}\backslash\{0\} \times\mathbb{R}\backslash\{0\}$ and  $(A,B)\in\mathbb{R}\backslash\{0\}\times(-1,1)\backslash\{0\}$, then
\begin{align}
    f\circ g(x,y)&=(x,y),\\
    g\circ f(A,B)&=(A,B).
\end{align}
Thus $f(A,B)$  and $g(x,y)$ are inverse of each other.
\end{proof}

Now we will explore the relations between Eqs. \eqref{eq: AB to xy}-\eqref{eq: xy to AB} when $x,y,A=0;B=0,\pm1$, as an extension to  Theorem \ref{theorem: SN-SP bijective}.
\begin{Definition}
    Suppose that $(x,y)\in \mathbb{R}\times \mathbb{R}$ and $(A,B)\in \mathbb{R}\times [-1,1]$.
    \begin{itemize}
        
        \item \textbf{Relation 1:} $(x,y)$ is related to $(A,B)$ via Eq. \eqref{eq: AB to xy} if $(x,y)$ and $(A,B)$ satisfy  Eq. \eqref{eq: AB to xy} i.e.,
        \[(x,y)=((1-|B|)A, BA) \]
        or equivalently $f(A,B)=(x,y)$ (Eq. \eqref{eq: f(A,B)});
        
        \item \textbf{Relation 2:} $(x,y)$ is related to $(A,B)$ via Eq. \eqref{eq: xy to AB} if $(x,y)$ and $(A,B)$ satisfy  Eq. \eqref{eq: xy to AB} i.e.,
        \[(A,B)=(\mathrm{sign}(x)(|x|+|y|),\mathrm{sign}(x)\frac{y}{|x|+|y|}) \]
        or equivalently $g(x,y)=(A,B)$ (Eq. \eqref{eq: g(x,y)})
        
    \end{itemize}
\end{Definition}
Relation 1 extends the bijectivity of $f$, and relation 2 extends the bijectivity of $g$, when $x,y,A\neq0;B\neq0,\pm1$.
We will explore each relation when $x,y,A=0;B=0,\pm1$, starting with relation 1
\begin{enumerate}
    
    \item if $x=0$,then
    \begin{itemize}
        \item $A=0\Rightarrow y=0,B=free$; or
        \item $B=\pm 1\Rightarrow y=\pm A,A=free$
    \end{itemize}

    \item if $y=0$,then
    \begin{itemize}
        \item $A=0\Rightarrow x=0,B=free$; or
        \item $B=0\Rightarrow x=A ,A=free$;
    \end{itemize}

    \item if $A=0$, then $x=y=0,B=free$;

    \item if $B=0$, then $x=A,y=0,A=free$;

    \item if $B=\pm 1$, then $x=0,y=\pm A,A=free$.
    
\end{enumerate}

For summary, we organize the above cases in tuples of order  pairs $((x,y),(A,B))$ that satisfies relation 1:
\begin{gather*}
    ((0,0),(0,b))\\
    ((0,0),(a,0))\\
    ((a,0),(a,0))\\
    ((0,ac),(a,c))
\end{gather*}
for all $a,b,c$ such that $a\in\mathbb{R}$, $b\in\mathbb[-1,1]$ and $c\in\{-1,1\}$. Observe that if $a\neq0$, then $f$ is not well defined function because the point $(A,B)=(a,0)$ in the domain of $f$ is mapped to two distinct points in the co-domain $(x,y)=(a,0)$ and $(x,y)=(0,0)$. If we exclude the order pair $((a,0),(a,0))$, then we have a well defined function $f$, but it is not injective because $f(a,0)=f(0,b)=(0,0)$ but $(a,0)=(0,b)$ for some $a,b\neq0$.

Similarly we will explore relation 2 when $x,y,A=0;B=0,\pm1$, noticing that $x$ and $y$ cannot simultaneously be zero
Since $B$ has an denominator $|x|+|y|$, we cannot have $|x|+|y|=0$:
\begin{enumerate}
    
    \item if $x=0$,then $A=B=0, y=free\neq0$; 
    
    \item if $y=0$,then $B=0,A=x,x=free\neq0$;

    \item if $A=0$,then $\mathrm{sign}(x)=0\Rightarrow x=0\Rightarrow B=0, y=free\neq0$;

    \item if $B=0$,then
    \begin{itemize}
        \item $\mathrm{sign}(x)=0\Rightarrow x=0\Rightarrow A=0, y=free\neq0$; or
        \item $y=0\Rightarrow A=x, x=free\neq0$;
    \end{itemize}
    \item if $B=\pm 1$,then
    \begin{align*}
        1&=|B|\\
        &=|\mathrm{sign}(x)|\frac{|y|}{|x|+|y|} 
    \end{align*}
    which is a contradiction, because $|\mathrm{sign}(x)|\neq 0$ (otherwise we have $1=0$), so we must have $|\mathrm{sign}(x)|=1$. Therefore $|x|>0$ and $1=\frac{|y|}{|x|+|y|}<1$. Therefore $B=\pm1$ is impossible.
    
\end{enumerate}

For summary, we organize the above cases in tuples of order  pairs $((x,y),(A,B))$ that satisfies relation 2:
\begin{gather*}
    ((0,r),(0,0))\\
    ((r,0),(r,0))\\
\end{gather*}
for all $r\neq0$. Observe that there is no solution to $g(x,y)=(A,\pm1)$, therefore $g$ is not surjective when including $B=\pm1$ in the co-domain.

\section{The Connection to the Möbius-Strip}\label{sec: The Connection to the Möbius-Strip}

A simple observation of the Soft Numbers coordinate system reveals a connection to the Möbius-strip. The two real axes are reciprocals of one another, and gluing the right boundary to the left boundary appears to generate an infinite Möbius-strip. However, this construction faces a twofold difficulty: first, an infinite strip cannot be realized in the physical world; second, the strip is too tall relative to its width. In this section we resolve the problem differently, and show how Soft Numbers can be linked to a geometric representation of the Möbius-strip. Firstly, we show the  equations, and then we show simple illustrations in Python to generate the SNS, the soft number's cartesian plane and the Möbius-strip.

\subsection{Möbius-Strip Equations with SNS}
For illustration purposes, we assume that the signed height $A$ is finite i.e., 
\begin{equation}
  |A| \le A_\textrm{max}  
\end{equation}
for some $A_\textrm{max}>1$ large enough. Observe that, the total height of the bounded SNS is 2$A_\textrm{max}$. If we take
\begin{equation}
  A_\textrm{max}=\pi R  
\end{equation}
for some $R>1$ large enough, then the parameter $R$ is the radius of the bending of the SNS (with circle's perimeter 2$\pi R$) into a Möbius-strip, and $A$ can be presented as a (signed) arc-length of the bent SNS so that 
\begin{equation} \label{eq: A dom for mobius}
\begin{split}
    A&=\phi R\in [-\pi R,\pi R],\\
    \phi&=\frac{A}{R}\in[-\pi,\pi].
\end{split}    
\end{equation}
In this setup, the arc-length $A$ and the bending angle $\phi$ are equivalent, up-to a constant multiple of the radius $R$ (in this work, we keep the radius $R$ constant). For Möbius-strip construction, the first variable is the arc-length $A$ (or equivalently the bending angle $\phi$). The other variable is the width of the Möbius-strip, that we define it from the SNS as
\begin{equation} \label{eq: B dom for mobius}
    B\in [-1,1].
\end{equation}

The soft number in the Cartesian plane is
\begin{equation}\label{eq: xy dom for mobius}
\begin{split}
  &C=x\bar{0}\softplus y,\\
  &x=(1-|B|)A,\\
  &y=BA,\\
  &|x|+|y| \le \pi R,
\end{split}    
\end{equation}
where the region in the last inequality is taken subject to the boundedness of $A$ being between $-\pi R$ and and $\pi R$ (cf. Eqs \eqref{eq: xy to AB} and \eqref{eq: A dom for mobius}). We denote by $X$, $Y$, $Z$ the cartesian coordinates of the Möbius-strip as follows:
\begin{equation}\label{eq: mobius strip eq}
\begin{split}
  &X=\left(R+B\cos\left(\frac{\phi}{2}\right)\right)\cos\left(\phi\right),\\
  &Y=\left(R+B\cos\left(\frac{\phi}{2}\right)\right)\sin\left(\phi\right),\\
  &Z=B\sin\left(\frac{\phi}{2}\right).\\
\end{split}    
\end{equation}
We use an RGB color model to color-code the domain of the SNS variables $A$ and $b$ as follows:

\begin{equation} \label{eq: AB color code}
C_\textrm{color} = \begin{cases}
1  \text{ (red)}& \text{if } A<0  \text{ and } B>0\\
0.7\text{ (yellow)}&  \text{if } A>0  \text{ and } B>0\\
0.5 \text{ (green)}&  \text{if } A<0  \text{ and } B<0\\
0 \text{ (blue)}  & \text{otherwise}
\end{cases}.
\end{equation}
With this RGB color coding, we will demonstrate how the SNS is mapped to a soft number cartesian plane and the Möbius-strip.

\subsection{Illustrations of Soft Numbers Cartesian Plane, SNS and Möbius-Strip}
The following Python code illustrates the soft numbers cartesian plane, the SNS and the Möbius-Strip, as we defined above for radius parameter $R=10$.

\begin{lstlisting}[language=Python, caption=Soft Numbers Cartesian Plane and SNS and Möbius-Strip in Python ]
import numpy as np
import matplotlib.pyplot as plt
from matplotlib import cm
plt.rc('text.latex', preamble=r'\usepackage{amsmath,amssymb,amsfonts}')
plt.rcParams['mathtext.fontset'] = 'stix'
plt.rcParams['font.family'] = 'STIXGeneral'

# SNS gererating
R = 10.0 #radius param
phi = np.linspace(-np.pi,np.pi,1000) #Angle var
B = np.linspace(-1,1,1000) #SNS width var
B, phi = np.meshgrid(B, phi)
A = phi*R #SNS height var
C = 0.7*(A>0)*(B>0)+0.0*(A>0)*(B<0)+0.5*(A<0)*(B<0)+1.0*(A<0)*(B>0) #color coding

# soft number components generating
x = (1-np.abs(B))*A #soft zero part
y = B*A #real part

# Mobius Strip vars generating
X = (R + B * np.cos(phi / 2)) * np.cos(phi)
Y = (R + B * np.cos(phi / 2)) * np.sin(phi)
Z = B * np.sin(phi / 2)

# plotting
C_norm = (C - C.min()) / (C.max() - C.min())
colors = cm.jet(C_norm)
## Mobius strip 
fig = plt.figure()
ax = fig.add_subplot(111, projection='3d')
p=ax.plot_surface(X, Y, Z, cmap='jet', facecolors=colors, shade=False)
ax.set_xlabel(r'$X$')
ax.set_ylabel(r'$Y$')
ax.set_zlabel(r'$Z$')
ax.set_box_aspect(None,zoom=0.85)
ax.axis('equal')
ax.set_title('Mobius Strip')
fig.colorbar(p,cax=None, ax=ax,label='Color Code')
## SNS
fig = plt.figure()
ax = fig.add_subplot(111, projection='3d')
p=ax.plot_surface(B, A, C, cmap='jet', facecolors=colors, shade=False)
ax.set_xlabel(r'$B$')
ax.set_ylabel(r'$A$')
ax.view_init(elev=90, azim=270)
ax.set_xlim([-1.5,1.5])
ax.set_zticks([])
ax.set_title('Soft Numbers Strip')
## soft number Cartesian Plane
fig = plt.figure()
ax = fig.add_subplot(111, projection='3d')
p=ax.plot_surface(x, y, C, cmap='jet', facecolors=colors, shade=False)
ax.set_xlabel(r'$x[\overline{0}]$')
ax.set_ylabel(r'$y$')
ax.view_init(elev=90, azim=270)
ax.set_zticks([])
ax.set_box_aspect(None,zoom=0.85)
ax.set_title('Soft Numbers Cartesian Plane')
\end{lstlisting}

\begin{figure}[ht]
    \centering
    \subfloat[\centering]{\includegraphics[width=5.0cm]{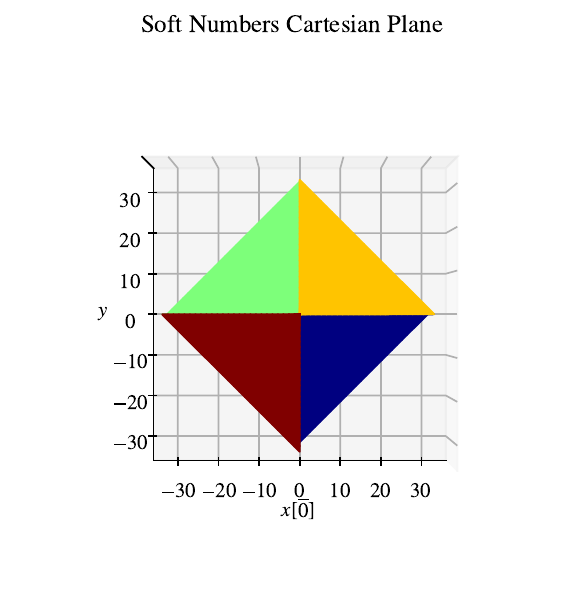}\label{fig: Soft Möbius Map (a)}}
    \subfloat[\centering]{\includegraphics[width=5.0cm]{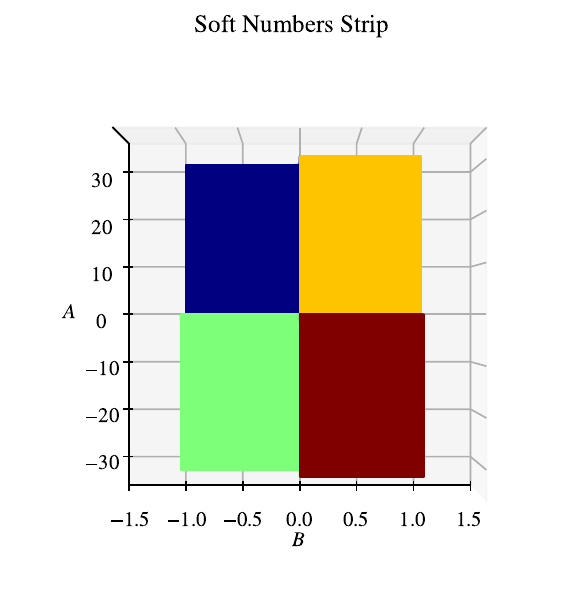}\label{fig: Soft Möbius Map (b)}}
    \subfloat[\centering]{\includegraphics[width=6.0cm]{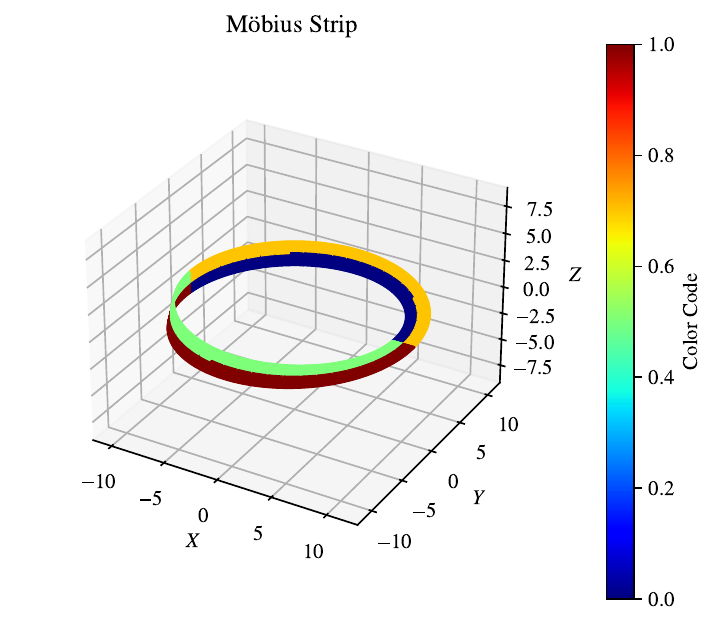}\label{fig: Soft Möbius Map (c)}}
    \caption{The Soft Möbius Map (with common color code bar): (a) Cartesian square domain of the soft numbers (colored); (b) Complete Soft coordinate system (colored); (c) Möbius-strip mapped by the soft numbers.}
    \label{fig: Soft Möbius Map}
\end{figure}

The illustrated results are shown in Figure \ref{fig: Soft Möbius Map}. We started with the SNS in Figure \ref{fig: Soft Möbius Map (b)}. The domain of the SNS is all $A,B$ such that $A\in[-10\pi,10\pi]$ (note: $10\pi\approx31.4$), and $B\in[-1,1]$. The SNS domain is divided into 4 colored rectangles according to Eq. \eqref{eq: AB color code}. Each colored rectangle in the SNS is, one the one hand, transformed into a triangle in the soft number cartesian plane in Figure \ref{fig: Soft Möbius Map (a)} (cf. Eq. \eqref{eq: xy dom for mobius}), and, on the other hand, transformed according to Eq. \eqref{eq: mobius strip eq} to generate the Möbius-strip shown in Figure \ref{fig: Soft Möbius Map (c)}.

We have demonstrated the connection between Soft Numbers and the Möbius-strip. In the context of a geometric understanding of the universe and physical phenomena, it is important to develop a mathematical language that incorporates the observer as an integral part of the world of phenomena. The Möbius-strip model is a geometric illustration of the idea that the observer is an organic part of the world itself, and in this way allows for a deeper understanding of the nature and character of Hilbert's sixth problem.

\section{The Probability of the Impossible}

Classical probability theory focuses primarily on computing the probability of possible events. In this paper, we show that probabilistic distinctions can be made even for impossible events.

\begin{Remark}
    A probability of a discrete random variable $X$ to have a value $x$ in the support $\Omega_X$ is denoted by
    \begin{equation}
    \textrm{Pr}(X=x)=p_X(x),
    \end{equation}
    where the RHS is denoted as the probability mass function (PMF) of the discrete random variable $X$. In the ordinary probability theory, if the value $x$ is outside of the support $\Omega_X$, then that probability is zero. In this paper, we extend those zero values into soft zeros.      
\end{Remark}

\subsection{Motivation} \label{ssec: impossible prob motivation}
Suppose we toss 6 fair coins, each with a 50\% chance of heads and 50\% chance of tails. The probability that the number of tails equals 4 is computed using the combinations formula. Let $X$ be the number of tails i.e., $X \sim  \text{Bin}\left(n=6,p=\frac{1}{2}\right)$. Recall that for general parameters $n\in\mathbb{N}$ and $p\in[0,1]$:
\begin{equation}
    \textrm{Pr}(X=k)=\binom{n}{k}p^k(1-p)^{n-k},
\end{equation}
for $k\in{0,1,2,...,n}$ as the number of tails. Take a "valid" numerical example for $k=4$ tails:
\begin{equation*}
    \textrm{Pr}(X=4)=\binom{6}{4}\cdot\frac{1}{2^6} = \frac{15}{64}.
\end{equation*}
The general calculation of this problem can be described using Pascal’s triangle (Figure \ref{fig: Pascal triangle}).
\begin{figure}[ht]
\centering
\begin{tikzpicture}
    \foreach \n in {0,...,5} {
        \foreach \k in {0,...,\n} {
            \node at (\k-\n/2, -\n) {
                \pgfmathparse{int(factorial(\n)/(factorial(\k)*factorial(\n-\k)))}
                \pgfmathresult
            };
        }
    }
\end{tikzpicture}
\caption{Pascal triangle.}
\label{fig: Pascal triangle}
\end{figure}

We now ask: what is the probability of getting tails 7 times when tossing 6 coins? Clearly the answer should be 0. Moreover, the mathematical expression for it is considered undefined:
\begin{equation*}
    \textrm{Pr}(X=7)=\binom{6}{7}\cdot\frac{1}{2^6}.
\end{equation*}
The probability of getting tails 8 times is also zero — but is it \textbf{exactly} the same probability as getting tails 7 times? It is commonly assumed, but we propose that there are different degrees of impossibility. To this end, we show that Pascal’s triangle can be extended to the right using soft numbers, allowing us to distinguish between the probabilities of impossible events.

\subsection{Factorial and Pascal triangle Extension by Soft Numbers}
The pascal triangle and the factorial were extended with soft numbers in \cite{Klein_Maimon2023}.

\subsubsection{Factorial Extension by Soft Numbers}
The factorial is extended with soft numbers, based on the observations that $0!=1$ and the factorial's recursion $(n-1)!=\frac{n!}{n}$ for $n\in\mathbb{N}$. This recursion was extended to non-positive integers as follows by replacing $n-1$ with $-m$ for $m\in\mathbb{N}$:
\begin{subequations}\label{eq: soft factorial}
\begin{align}
   (-m)!&=\frac{(-(m-1)!)}{-(m-1)}=(-1)^{m-1}\frac{1}{(m-1)!\bar{0}}\label{eq: soft factorial formula},\\
   (-1)!&=(-1)^0\frac{1}{(0)!\bar{0}}=\frac{1}{\bar{0}},\label{eq: soft factorial init.}    
\end{align}
\end{subequations}
where the middle quantity in Eq. \eqref{eq: soft factorial formula} is the recursion of the LHS, the RHS of Eq. \eqref{eq: soft factorial formula} is the direct formula of the LHS based on the regular factorial in the denominator, and Eq. \eqref{eq: soft factorial init.} is the initial condition of the recursion.

\subsubsection{Binomial Coefficient Extension by Soft Numbers}
The binomial coefficient is extended based on the extension of the factorial and the ordinary definition of the binomial coefficient:
\begin{equation}\label{eq: bin coeff soft k>n>=0}
    \binom{n}{k}=\frac{n!}{k!(n-k)!}=(-1)^{k-n-1}\frac{n!(k-n-1)!}{k!}\bar{0},
\end{equation}
The RHS is obtained if $k>n\ge 0$.
\begin{Remark}
    The reduction formula of the binomial coefficient 
    \begin{equation}
     \binom{n+1}{k+1}=\binom{n}{k}+\binom{n}{k+1}   
    \end{equation}
    holds for $k>n\ge 0$ (in addition to $n>k\ge 0$). See the theorem and the proof in \cite[p. 131]{Klein_Maimon2023}.   
\end{Remark}
\begin{Theorem}
Suppose $n,m\in\mathbb{N}$. Then
\begin{equation}\label{eq: bin coeff soft n,m>=0 thm}
    \binom{n}{-m}=\binom{n}{n+m}
\end{equation}
i.e., the extension of the lower position $-m$ is well defined for negative number. 
\end{Theorem}
\begin{proof}
    The LHS is
    \begin{align*}
     \binom{n}{-m}&=\frac{n!}{(-m)!(n-(-m))!}\\
     &=\frac{n!}{(-m)!(n+m)!}.
    \end{align*}
    The RHS is
    \begin{align*}
     \binom{n}{n+m}&=\frac{n!}{(n+m)!(n-(n+m))!}\\
     &=\frac{n!}{(n+m)!(-m)!}.
    \end{align*}
    We have equality between the LHS and the RHS, and thus the extension of the lower position $-m$ is well defined for negative numbers.
\end{proof}
\begin{Corollary}
Suppose $n,m\in\mathbb{N}$. Then
\begin{equation}\label{eq: bin coeff soft n,m>=0 cor}
    \binom{n}{-m}=(-1)^{m-1}\frac{n!(m-1)!}{n+m!}\bar{0}.
\end{equation}
\end{Corollary}
\begin{proof}
    Immediate result using Eq. \eqref{eq: bin coeff soft k>n>=0} (with $k=n+m$) and Eq. \eqref{eq: bin coeff soft n,m>=0 thm}.
\end{proof}

\subsubsection{Illustration of Extended Pascal triangle}
An illustration of the new Pascal triangle, extended by soft numbers is shown in Figure \ref{fig: The extended Pascal triangle}.
\begin{figure}[htb]
    \centering
    \includegraphics[width=0.5\linewidth]{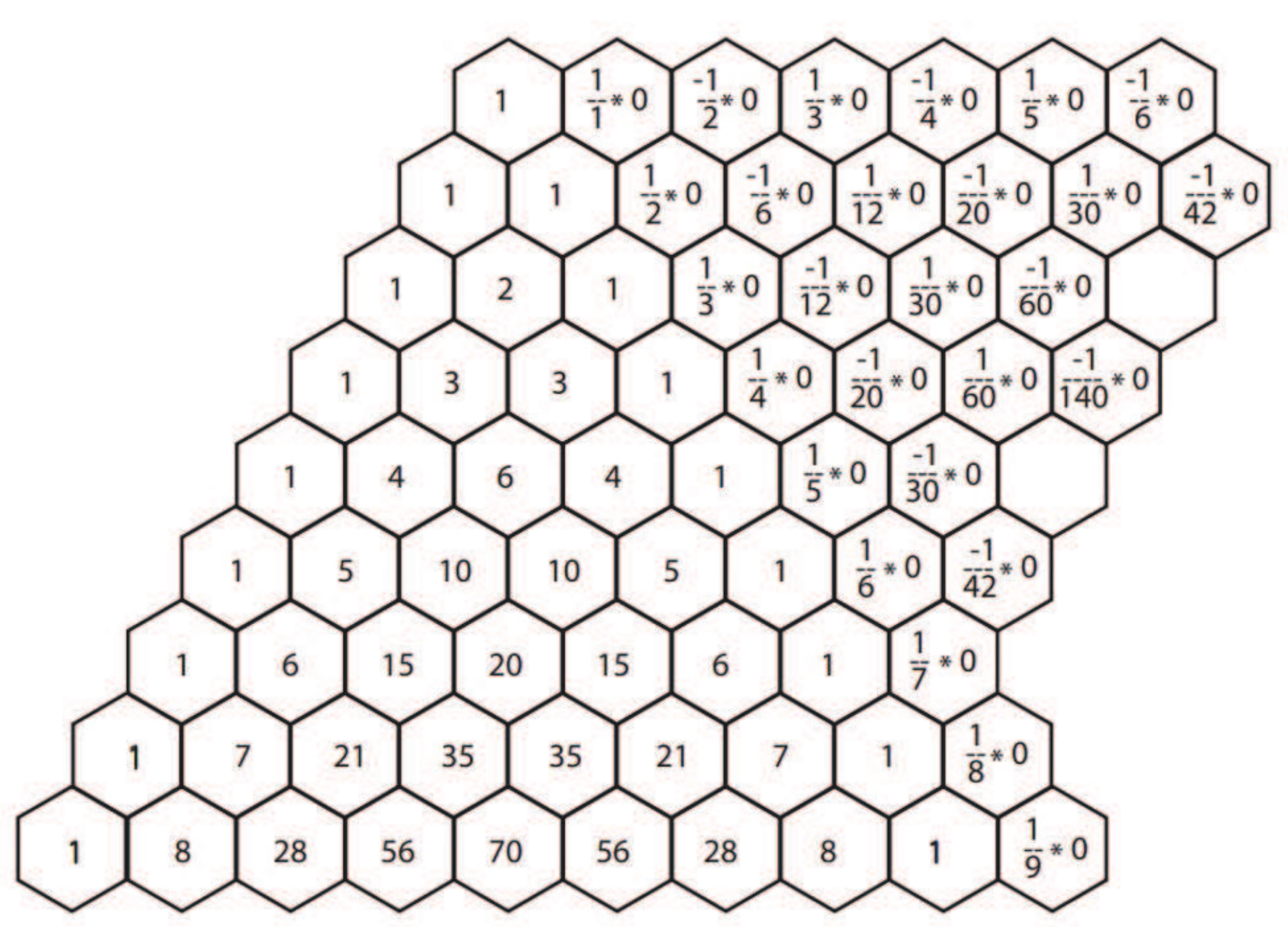}
    \caption{The extended Pascal triangle\cite[Fig. 5]{Klein_Maimon2023}}
    \label{fig: The extended Pascal triangle}
\end{figure}
In this figure, the regular triangle is shown in the middle and padded the real values of the binomial coefficients $\binom{n}{k}$ with $n>k \ge 0$. The extension on the right edge of the Pascal triangle is illustrated in this figure, and the values are corresponded to the binomial coefficients $\binom{n}{k}$ with $k>n \ge 0$. This extension, can be accomplished symmetrically on the left edged for $k<0 \le n$.  

\subsection{Impossible Probability - Definitions and Notations}
In order to distinguish among the notations of probability operators ($\mathrm{Pr}$ and $\mathrm{Ps}$), we denote new symbols for impossible probability in the case of a discrete random variable.

\begin{Definition}\label{def: impossible prob}
    Let $X$ be a discrete random variable with a support $\Omega_X$ and a PMF $p_X(\cdot)$. We denote a new symbol $\mathrm{Pi}$ as the \emph{Impossible Probability} operator for a probability of impossible events, and a new symbol for a $p_{\textrm{i}_X}(\cdot)$ for a soft zero extension of the PMF $p_X(\cdot)$. We define the probability of impossible events for discrete random variables as follows:
    \begin{equation}\label{eq: Pi(X=x)}
    \mathrm{Pi}(X=x) = \begin{cases} 
      p_{\textrm{i}_X}(x), & x \notin \Omega_X \\
      0\bar{0}, &  x \in \Omega_X 
   \end{cases}
    \end{equation}
\end{Definition}
In the context of Definition \ref{def: impossible prob}, $\mathrm{Pi}(X=x)$ is a soft zero, such that the PMF soft zero extension is $p_{\textrm{i}_X}(x)=f(x)\bar{0}$ for some real function $f$. If $x$ is inside the support $\Omega_X$, then $X=x$ is a \emph{possible} event, and thus we map it to $0\bar{0}$ for convenience. Now we denote another symbol that merges the impossible and possible events into one definition.
\begin{Definition}\label{def: impossible possible prob}
    Let $\mathrm{Pip}(\cdot)$ be the \emph{Impossible/Possible Probability} operator. For a discrete random variable $X$, we define $\mathrm{Pip}(\cdot)$ by
    \begin{equation}\label{eq: Pip(X=x)}
       \mathrm{Pip}(X=x)=\mathrm{Pi}(X=x)\softplus\mathrm{Pr}(X=x) 
    \end{equation}
\end{Definition}
In the context of Definition \ref{def: impossible possible prob}, $\mathrm{Pip}(X=x)$ is a soft number that cannot have non-zero components simultaneously in both parts. If $x$ is inside the support $\Omega_X$, then 
\begin{equation*}
    \mathrm{Pip}(X=x)=0\bar{0}\softplus p_X(x),
\end{equation*}
i.e., it has zero in ths soft part and the PMF of $X$ in the real part. If $x$ is outside of the support $\Omega_X$, then 
\begin{equation*}
    p_{\textrm{i}_X}(x)\softplus 0,
\end{equation*}
i.e., it has the extended PMF $p_{\textrm{i}_X}(x)=f(x)\bar{0}$ in the soft part and zero in the real part.

\subsection{Examples}
We show the following examples for demonstrations of the probability of the impossible/possible with the above definitions.
\begin{Example}[Binomial Distribution]
     Now we return to the binomial distribution in subsection \ref{ssec: impossible prob motivation}, $X \sim  \text{Bin}\left(n,p\right)$, and let $k\in \mathbb{N}$. If $n>k\ge0$, then denote $k=m+n$ for 
     \begin{equation*}
          \mathrm{Pip}(X=k)=0\bar{0}\softplus \frac{n!}{k!(n-k)!}p^k(1-p)^{n-k},    
     \end{equation*}
\end{Example}
\noindent where the real part is a result of the ordinary binomial distribution with an ordinary and real binomial coefficient. If $k>n\ge0$, then denote $k=n+m$ for some $m\in \mathbb{N}$, and we have
     \begin{equation*}
          \mathrm{Pip}(X=n+m)=(-1)^{m-1}\frac{n!(m-1)!}{(n+m)!}p^{n+m}(1-p)^{-m}\bar{0}\softplus0.    
     \end{equation*}
If $k<0<n$, then denote $k=-m$ for some $m\in \mathbb{N}$, and we have
     \begin{equation*}
          \mathrm{Pip}(X=-m)=(-1)^{m-1}\frac{n!(m-1)!}{(n+m)!}p^{-m}(1-p)^{n+m}\bar{0}\softplus0.    
     \end{equation*}
We have some symmetry between the last two case: The binomial coefficient is identical as expected, but the exponent of $p$ and $(1-p)$ a are swapped.

In Table \ref{tab: impossible/possible probs binomial distribution Bin(n=6,p=1/2)}, we show some numerical examples of impossible/possible probabilities for $n=6$ and $p=\frac{1}{2}$ the the binomial distribution. Since $p=\frac{1}{2}$, every pair of integers ($k>6$,$k<0$) that sum to $n=6$ share the same impossible/possible probability value. In addition, as $k$ diverges from the set \{0,1,...,6\} to infinity (and negative infinity), the absolute value of the impossible probability part  monotonically converges to zero i.e., it become more impossible. 

\begin{table}[htb]
    \centering
    \begin{tabular}{c|c}
        $k$ & $\mathrm{Pip}(X=k)$ \\
        \hline
         $4$ & $0\bar{0}\softplus\frac{15}{64}$ \\
         $7$ & $\frac{1}{448}\bar{0}\softplus0$ \\
         $8$ & $-\frac{1}{35854}\bar{0}\softplus0$ \\
         $-1$ & $\frac{1}{448}\bar{0}\softplus0$ \\
         $-2$ & $-\frac{1}{35854}\bar{0}\softplus0$ \\
    \end{tabular}
    \caption{Numerical examples of impossible/possible probabilities for binomial distribution $X\sim\text{Bin}\left(n=6,p=\frac{1}{2}\right)$}
    \label{tab: impossible/possible probs binomial distribution Bin(n=6,p=1/2)}
\end{table}

\begin{Example}[Poisson Distribution]
Let $X$ be a discrete random variable with Poisson distribution with parameter $\lambda$, $X\sim\text{Pois}(\lambda)$. The probability of obtaining $k$ events is given by
\begin{equation}
  \mathrm{Pr}(X=k)=\frac{\lambda^k}{k!}e^{-\lambda}.  
\end{equation}
The extension with impossible probability is accomplished by considering $k<0$. In this case, we have the following impossible/possible probability is (replace $k=-m$, $m\in\mathbb{N}$)      \begin{equation*}
        \mathrm{Pip}(X=-m)=(-1)^{m-1}(m-1)!\lambda^{-m}e^{-\lambda}\bar{0}\softplus0.    
     \end{equation*}    
\end{Example}
We have thus shown that soft numbers can be used to compute impossible events in the Poisson distribution.



\section{Discussion} \label{sec: Discussion}
In the paper, we presented a new approach to classical probability theory by distinguishing between two events whose probability is zero. When Kolmogorov published probability theory, the scientific world was in turmoil over the probabilistic interpretation of quantum theory. In fact, Kolmogorov developed the theory in accordance with the model developed by Pascal and Fermat in the 17th century. Our approach extends classical probability theory by distinguishing between the different multiples of the number zero. We defined a new type of numbers called soft numbers.

Algebraically, soft numbers are similar to the dual numbers developed by Clifford \cite{Clifford1871}. The difference lies in the distinction that the infinitesimal space arises directly from the creation of a space of infinitely many zeros that are distinct from one another, and in the geometric interpretation of the zero axis. This discovery allows us to create a new coordinate system that is non-Cartesian and produces a simple geometric model of a Möbius strip. In doing so, we link the two parts of Hilbert's sixth problem. Moreover, a set, soft numbers represent better the two possible matrix representations of the dual numbers than the traditional dual numbers, because the dual part $\varepsilon$ usually represent only one option between two matrix representations in $M_2(\mathbb{R})$, while soft numbers includes two bridge numbers with the same nilpotent part, the soft zero part of the soft number. In addition, we established and extension of the soft numbers by polynomial rings, quotient rings of quadratic algebra, and more specific of the complex numbers, the split-complex numbers and the dual numbers. In some sense, the dual numbers are part of the soft numbers.

A significant recent contribution to Hilbert's sixth problem was made by Deng et al. \cite{Deng2025}, who rigorously derived the fundamental equations of fluid mechanics starting from Newtonian hard-sphere particle systems, using Boltzmann's kinetic theory as an intermediate step. At its core, their result addresses a deep passage from local microscopic information to global macroscopic laws. The geometric framework developed in the present paper, based on Soft Numbers and the Möbius-strip structure, may be viewed as a complementary perspective on this local-to-global passage. The topological properties of the Möbius-strip-locally two-sided yet globally one-sided — provide a natural geometric model for the kind of emergent macroscopic behavior that Hilbert's sixth problem seeks to axiomatize. Developing a rigorous connection between these two approaches remains an open and promising direction for future research.

Using soft numbers, we are able to distinguish between two events that are impossible. In other words, we have discovered that there are different degrees of impossibility. In addition, we find that it is necessary to introduce an additional axiom to the five axioms of Kolmogorov in probability theory. The sixth axiom states that there exist certain events whose probability is an infinitesimal quantity.

\section{Conclusion}\label{sec: Conclusion}
Today there is a need to understand the laws of physics through the recognition that the observer is an organic part of the world of phenomena. A Möbius-strip, which has only one side, allows us to see how a person observes the world and also observes themselves as an observer of the world. In Deng et al. \cite{Deng2025}, the connection between local macroscopic phenomena and global macroscopic phenomena was investigated. A Möbius-strip expresses in its very essence a paradoxical relationship between locality and globality: locally it has two sides, but globally it has only one. We propose to enrich probability theory by introducing a structural distinction between events of zero probability, so that zero is no longer a single point but a space with internal structure. This structure makes it possible to define different degrees of impossibility, analogous to distinctions between the possible and the impossible in other areas of mathematics, such as constructions with a straightedge and compass, the resolvability of algebraic equations, or problems of computability, thereby revealing information that is lost within the classical framework of probability theory.

\bigskip
\noindent
\textbf{Soft Logic and Soft Numbers may therefore constitute not only an extension of probability theory, but also a new geometric language for the foundations of physics.}



\end{document}